\documentclass{article}
\usepackage{ijcai22}

\usepackage{times}
\usepackage{soul}
\usepackage{url}
\usepackage[hidelinks]{hyperref}
\usepackage[utf8]{inputenc}
\usepackage[small]{caption}
\usepackage{graphicx}
\usepackage{amsmath}
\usepackage{amsthm}
\usepackage{amsfonts}
\usepackage{booktabs}
\usepackage{algorithm}
\usepackage{algorithmic}
\usepackage{booktabs}
\usepackage{multirow}
\usepackage{multicol}
\usepackage{color}
\usepackage{wrapfig}
\usepackage{hyperref}
\newtheorem{example}{Example}
\newtheorem{theorem}{Theorem}
\newtheorem{definition}{Definition}

\newtheorem{proposition}{Proposition}

\newtheorem{corollary}{Corollary}
\newcommand{\thistheoremname}{}
\newtheorem*{genericthm*}{\thistheoremname}
\newenvironment{namedthm*}[1]
  {\renewcommand{\thistheoremname}{#1}%
   \begin{genericthm*}}
  {\end{genericthm*}}

\title{Strategyproof Mechanisms for Group-Fair Facility Location Problems}

\author{
Houyu Zhou$^1$
\and
Minming Li$^1$\and
Hau Chan$^2$
\affiliations
$^1$City University of Hong Kong\\
$^2$University of Nebraska-Lincoln\\
\emails
houyuzhou2-c@my.cityu.edu.hk,
minming.li@cityu.edu.hk,
hchan3@unl.edu
}
\begin{document}

\maketitle

\begin{abstract}
    We study the facility location problems where agents are located on a real line and divided into groups based on criteria such as ethnicity or age. 
    Our aim is to design mechanisms 
    to locate a facility to approximately minimize the costs of groups of agents to the facility fairly while eliciting the agents' locations truthfully. 
    We first explore various well-motivated group fairness cost objectives for the problems and show that many natural objectives have an unbounded approximation ratio. 
    We then consider minimizing the maximum total group cost and minimizing the average group cost objectives. 
    For these objectives, we show that 
    existing classical mechanisms (e.g., median) and new group-based mechanisms provide bounded approximation ratios, 
    where the group-based mechanisms can achieve better ratios. 
    We also provide lower bounds for both objectives. 
    To measure fairness between groups and within each group, we study a new notion of intergroup and intragroup fairness (IIF) . 
    We consider two IIF objectives 
    and provide mechanisms with tight approximation ratios. 
\end{abstract}

\section{Introduction}
In the classical facility location problems (FLPs) from the mechanism design perspective \cite{Procaccia:2009aa}, 
we have a set $N$ of agents, and each agent $i$ has a privately known location $x_i \in \mathbb{R}$ on the real line. 
The typical %mechanism design 
goal is to design a strategyproof mechanism %(without payment) 
that elicits true location information from the agents 
and locates a facility (e.g., a public library, park, or representative) $y \in \mathbb{R}$ 
on the real line to (approximately) minimize a given %(e.g., total or maximum) 
\emph{cost objective} that measures agents' distances to the facility. %(e.g., $\sum_{i \in N} |y - x_i|$ or $\max_{i \in N} |y - x_i|$).
%and locates a facility to (approximately) optimize the given cost objective. 
%Such a basic setting is well-understood by now %-- 
%some mechanisms (e.g., median and endpoint) are strategyproof 
%and approximately optimal for various cost objectives. 
Subsequent works have explored settings with more than 
one facility (e.g., \cite{Procaccia:2009aa,Lu:2010aa,Fotakis:2014aa,paolo2015heterogeneous,DBLP:conf/aaim/ChenFLD20,DBLP:conf/atal/ZouL15}) and 
various models/objectives (e.g., \cite{Sui:2013aa,Sui:2015aa,Filos-Ratsikas:2017aa,Cai:2016aa,feldman2013strategyproof,limei2016newobjective,Aziz:2019aa,Aziz:2020aa}). 
%capacity constraints \cite{Aziz:2019aa,Aziz:2020aa}. %, and 
%automated mechanism design \cite{Narasimhan:2016aa,Golowich:2018aa}. 
See \cite{survey} for a  survey. % on the topic. 

Motivated by the importance of ensuring group fairness and equality among groups of agents in our society, 
\begin{wrapfigure}{r}{0.20\textwidth} %this figure will be at the right
\vspace{-0.2cm}
\hspace{-0.7cm}
    \centering
    \includegraphics[width=0.23\textwidth]{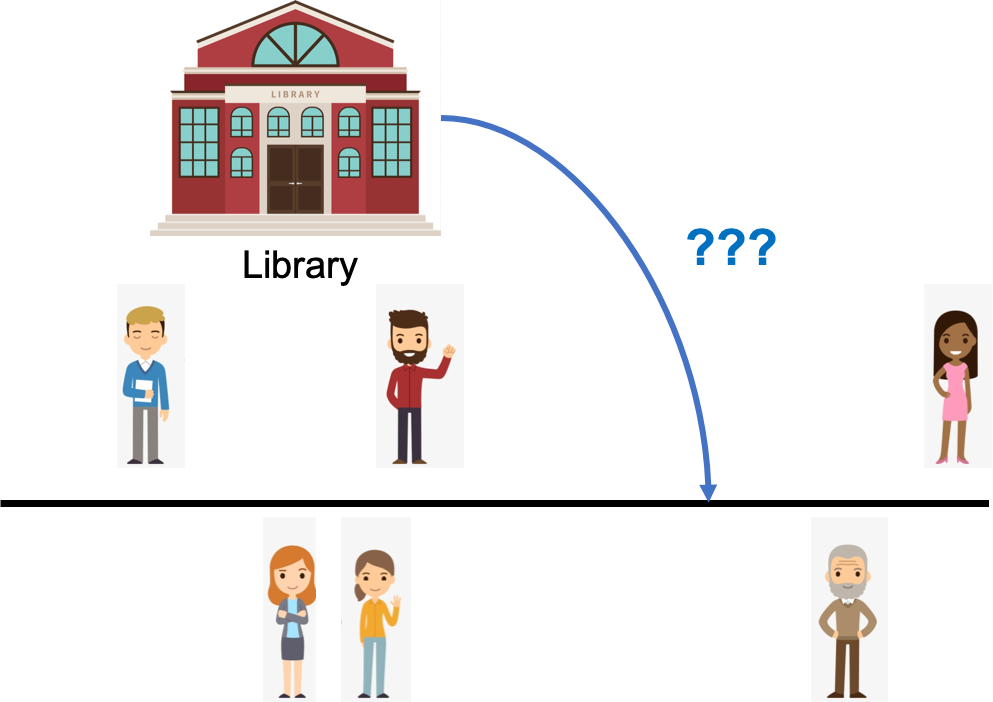}
      \captionsetup{oneside,margin={-0.5cm,0cm}}
    \caption{Group-Fair FLPs}
    \label{fig:flps}
    \vspace{-0.5cm}
\end{wrapfigure}
%(e.g., legal definitions of discrimination \cite{Barocas:2016aa,Zimmer:1996aa,Rutherglen:1987aa} and recent social right movements), 
we consider the \emph{group-fair FLPs} (see Figure \ref{fig:flps}) where the set $N$ of $n$ agents is partitioned into $m$ groups, $G_1, ..., G_m$,
based on criteria (e.g., gender, race, or age) 
and aim to design strategyproof mechanisms to locate the facility to serve 
groups of agents to ensure some desired forms of group fairness and the truthfulness of agents.  
Potential applications include determining the best location of a public facility (e.g., park or library) to serve a subset of agents  
as to provide fair access to different groups (e.g., based on race, gender, and age) 
and determining the best candidate or ideal point to select in an election within an organization 
as to ensure fair representative for groups of agents. 
%are typical examples the social planner encounters frequently. 

\paragraph{Our Key Challenges and Contribution.}
%We totally agree that Mechanism 2 (putting the facility at the median of the largest group) can be unfair for other groups of agents. After further consideration of your concerns, 
It is easy to see that any facility location (e.g., either at a fixed point or at an agent's location) 
can be unfair for different groups of agents under some FLP instances. 
Therefore, in order to understand the degree of unfairness suffered by each group of agents, 
we use group-fair objectives to mathematically quantify and compare the unfairness (or inequality) 
of each group (e.g., based on group costs) given the location of the facility. 
%In this paper, we consider  
%On the other hand, 
%the previously unexplored \emph{group-fair facility location} problems in the approximate mechanism design without money setting. %paradigm of approximate mechanism design without money  
%In such a setting, 
Given the group-fair cost objectives, we aim to design strategyproof mechanisms  
to locate a facility that ensures the costs of groups are fair and can (approximately) minimize the overall unfairness of all groups. 
%(as to provide group-fair access to the facility) 
%subject to a given objective. %, 
%such a \emph{group-fair facility location} setting 
%has not been explored in the past. %even in the most basic facility location game. 
%Thus, this calls for an indispensable need of a new paradigm that addresses group fairness 
%in general approximate mechanism design without money. 
As a result, our key challenges are  %we consider the group-fair facility location problems in the approximate mechanism design without money setting 
%we address the following key questions. %using facility location games as an important case study. 
(1) defining appropriate group-fair objectives 
and (2) designing strategyproof mechanisms to (approximately) optimize a given group-fair objective. 

As our starting point, we consider group-fair cost objectives 
introduced by well-established domain experts \cite{marsh1994equity} 
in the optimization literature of FLPs dated back in the 1990s. 
%In the optimization literature of FLPs dated back in the 1990s, 
%well-established domain experts \cite{} have introduced many group-fair (or inequality) cost objectives. 
%introduced by well-established domain experts [Marsh and Schilling, 1994] dated back in the 1990s. 
As we showed (in Theorem \ref{thm:first3unbounded}), 
it is impossible to design any reasonable strategyproof mechanisms with 
a bounded approximation for many of these objectives. 
Despite the negative result, we identify 
other well-motivated group-fair cost objectives 
and introduce novel objectives that capture 
\emph{intergroup} and \emph{intragroup} fairness (IIF) for agents that are within each group, 
in which the design of strategypoof mechanisms with constant approximation ratios is possible. 
Using simple existing mechanisms and new mechanisms designed by us,   
our results (through rigorous non-trivial arguments) 
demonstrate the possibility of achieving approximately group-fair outcomes for these objectives. 
%Thus, we focus on designing mechanisms for the main objectives studied in the paper. 
%It is surprising to us that we were able to derive strategypoof mechanisms with constant approximation ratios 
%despite the negative result from Theorem 1 for almost all of the objectives proposed by the domain experts. 
%We also introduce novel objectives that capture 
%\emph{intergroup} and \emph{intragroup} fairness (IIF) for agents that are within each group. 
Our results are summarized in Table \ref{tab:Summary}. 
More specifically: 

\begin{table}[htb!]
    \centering
    \begin{tabular}{c|c|c|c}
    \toprule
        Mechanism & UB ($mtgc$) & UB ($magc$) & LB\\
    \midrule
        MDM & $\Omega(m)$ & $3$ &\multirow{3}*{2} \\
        LDM & $\Omega(n)$ & $\Omega(n)$ \\
        MGDM & $3$ & $3$  & \\
        \hline
        RM & $\Omega(n)$ & $\Omega(n)$ &\multirow{2}*{1.5}\\
        NRM & $\Omega(n)$ & $2$ & \\
        \bottomrule
        \toprule
        Mechanism & UB ($IIF_1$) & UB ($IIF_2$) & LB\\
        \hline
        k-LDM & \multicolumn{2}{c|}{$4$} & 4\\
    \bottomrule
    \end{tabular}
    \caption{Summary of Results, where UB means upper bound, LB means lower bound. See Sections 2-4 for details.}%of mechanisms and objectives.}
    \label{tab:Summary}
\end{table}

\begin{itemize}
    %\item We show that not every fairness objective has a bounded approximation ratio.   
    \item For the group-fair objective that aims to minimize the maximum total group cost ($mtgc$), 
    we show that the classical mechanisms (i.e., Median Deterministic Mechanism (MDM), Leftmost Deterministic Mechanism (LDM), and Randomized Mechanism (RM)) proposed by \cite{Procaccia:2009aa} 
%    , where the first two mechanisms put the %facility at the median and the leftmost agent %location, respectively, the third mechanism puts %the facility at the leftmost, the rightmost, and %the middle of the two points with some %probabilities, 
    have parameterized approximation ratios. 
    In contrast, the mechanism we proposed, Majority Group Deterministic Mechanism (MGDM), that leverages group information and puts the facility at the median of the largest group, obtains an improved approximation ratio of 3. 
    We also provide a lower bound of 2 for this objective.  
    We show that putting the facility at the median of any group with a smaller size will result in a larger approximation ratio. %{\color{red} or even strategyproofness}. 
    \item For the objective of minimizing the maximum average group cost ($magc$), we provide the upper bounds for all of the  considered mechanisms and the lower bounds for the objective. The Narrow Randomized Mechanism (NRM), which leverages randomization and group information, modifies the range used by RM for placing facilities to achieve a better approximation ratio.
    \item We introduce a new notion of intergroup and intragroup fairness (IIF) 
    that considers group-fairness among the groups and within groups. 
    We consider two objectives based on the IIF concept: $IIF_{1}$, $IIF_{2}$,     where the first considers these two fairness measures as two separate indicators, while the second considers these two fairness measures as one combined indicator for each group.
    For both objectives, we establish upper bounds of 4 by using the kth-Location Deterministic Mechanism (k-LDM), which puts the facility at the $k^{th}$ (leftmost) agent location, and the matching lower bounds of $4$. %{\color{red} we need to say what k-LDM is}
    %It is interesting to see that when one only considers intragroup fairness, only additive approximation is possible \cite{Cai:2016aa}. 
    %However, when we combine it with intergroup fairness, a multiplicative approximation can be obtained.
\end{itemize}

% Due to the space limit, most of the proofs are in the full version of the paper (\href{https://www.google.com/}{https://www.google.com/}).

%{\color{red} I think we need to say what these mechanisms are if they are not obvious}

\paragraph{Related Work.}
We will elaborate below the most related works of facility location problems (FLPs) to ours that consider 
some forms of individual fairness and envy (i.e., a special case of group fairness). 
%Although previous works have not considered the notion of group fairness explicitly, 
%a few work have considered some form of individual fairness and envy in general. 
From the optimization perspective, 
early works (see e.g., \cite{mcallister1976equity,marsh1994equity,Mulligan:1991ug}) in FLPs 
have examined objectives that quantify various inequity notions. 
For instance, 
%For the group fairness in FLPs,
\cite{marsh1994equity} considers a group-fairness objective (i.e., the Center objective in \cite{marsh1994equity}) 
that is equivalent to our $mtgc$ group-fair objective and similar to our $magc$ group-fair objective.  
%studies the measurements of group fairness, our $mtgc$ is equivalent to their measurement (1): 
%Center, and $magc$ is also based on this idea, but they does not study approximate strategyproof mechanism design in this problem.
These works do not consider FLPs from the mechanism design perspective. 

Although previous works have not considered the general notion of group fairness explicitly (e.g, $|G_j| > 1$ for some $j$), 
a few works have considered some form of individual fairness 
(i.e., the maximum cost objective when $m=n$) and envy in general. %(e.g., see \cite{Cai:2016aa}).} 
The seminal work of approximate mechanism design without money in FLPs \cite{Procaccia:2009aa} 
considers the design of strategyproof mechanisms that approximately minimize certain cost objectives such as total cost or maximum cost. 
%In the seminal work of approximate mechanism design without money in FLPs \cite{Procaccia:2009aa}, 
An individual fair objective that is considered in \cite{Procaccia:2009aa} is that of the maximum cost objective, 
which aims to minimize the maximum cost of the agent farthest from the facility (i.e., $m=n$). 
For the maximum cost objective, \cite{Procaccia:2009aa} establish tight upper and 
lower bounds of $2$ for deterministic mechanisms and $1.5$ for randomized mechanisms. 
However, applying these mechanisms to some of our objectives directly, such as the $mtgc$ and randomized part in the $magc$, 
would yield worse approximation ratios. %if we apply them in 
There are also envy notions such as minimax envy \cite{Cai:2016aa,DBLP:conf/aaim/ChenFLD20}, 
which aims to minimize the (normalized) maximum difference between any two agents' costs, and envy ratio \cite{DBLP:conf/aaim/LiuDCFN20,ding2020facility}, 
which aims to minimize the maximum over the ratios between any two agents' utilities. 
Other works and variations on facility location problems can be found in a recent survey \cite{survey}.
%We will elaborate below the most related works of FLPs to ours that consider some forms of individual fairness and envy (i.e., a special case of group fairness). 
%The seminal work \cite{Procaccia:2009aa} studies 
%the maximum cost objective aiming to minimize the maximum cost 
%of the the agent farthest from the facility (i.e., $m=n$). They establish tight upper and lower bounds of $2$ for deterministic and $1.5$ for randomized mechanisms, but these will achieve worse approximation ratios if we apply them in some of our objectives directly, like $mtgc$ and randomized part in $magc$. In terms of envy, there are envy notions such as minimax envy \cite{Cai:2016aa,DBLP:conf/aaim/ChenFLD20}, which aims to minimize the (normalized) maximum difference between any two agents' costs, and envy ratio \cite{DBLP:conf/aaim/LiuDCFN20,ding2020facility}, which aims to minimize the maximum over the ratios between any two agents' utilities.

\section{Preliminary}
In this section, we define group-fair facility location problems and consider several group-fair social objectives. We then show that some of these objectives have unbounded approximation ratios. 

\subsection{Group-Fair Facility Location Problems}
Let $N=\{1,2,...,n\}$ be a set of agents on the real line
and $G=\{G_1, G_2,...,G_m\}$ be the set of (disjoint) groups of the agents. 
Each agent $i \in N$ has the profile $r_i=(x_i, g_i)$ where $x_i \in \mathbb{R}$ is the location of agent $i$ and $g_i \in \{1, ..., m\}$ is the group membership of agent $i$. Each agent knows the
mechanism and then reports her location, which may be different
from her true location. We use $|G_j|$ to denote the number of the agents in group $G_j$. Without loss of generality, we assume that $x_1\leq x_2 \leq ...\leq x_n$. A profile $r=\{r_1,r_2,..,r_n\}$ is a collection of the location and group information. 
A deterministic mechanism is a function $f$ which maps profile $r$ to a facility location $y \in \mathbb{R}$. A randomized mechanism is a function $f$, which maps profile $r$ to a facility location $Y$, where $Y$ is a set of probability distributions over $\mathbb{R}$. 
Let $d(a,b)=|a-b|$ be the distance between two points $a, b \in \mathbb{R}$. 
Naturally, given a deterministic (or randomized) mechanism $f$ and the profile $r$, the cost of agent $i \in N$ is defined as $c(r,x_i)=d(f(r),x_i)$ 
(or the expected distance $\mathbb{E}_{Y\sim f(r)}[d(Y, x_i)]$). 

Our goal is to design mechanisms that enforce truthfulness while (approximately) optimizing an objective function. 

\begin{definition} \label{def:sp}
    A mechanism $f$ is \textbf{strategyproof (SP)} if and only if an agent can never benefit by reporting a false location, regardless of the strategies of the other agents. More formally, for any profile $r=\{r_1,...,r_n\}$, for any $i\in N$ and for any $x_i'\in \mathbb{R}$, let $r_i'=(x_i', g_i)$. We have $c(f(r), x_i) \leq c(f(r_i', r_{-i}), x_i)$ 
    where $r_{-i}$ is the profile of all agents except agent $i$. %{\color{red} Cannot assume other people are telling the truth. this r still looks like the true profile}
\end{definition}

Notice that when $f$ is randomized, Definition \ref{def:sp} implies SP in expectation.
%\color{red} missing a randomized truthful definition; a reviewer asked for that; footnote it or add a sentence}

In the following, we discuss several group-fair cost objectives that model some form of inequity.
We invoke the legal notions of disparate treatment \cite{Barocas:2016aa,Zimmer:1996aa} 
and disparate impact \cite{Barocas:2016aa,Rutherglen:1987aa}, 
the optimization version of FLPs \cite{mcallister1976equity,marsh1994equity,Mulligan:1991ug},  
and recent studies in optimization problems \cite{Tsang:2019aa,Celis:2018aa} to derive and motivate the following group-fair objectives. 

\paragraph{Group-fair Cost Objectives.} 
We consider defining the group cost from two main perspectives. 
One is the total group cost, which is the sum of the costs of all the group members. 
Through the objective, we hope to measure the inequality of each group in terms of its group cost to the facility.  
Hence, our first group-fair cost objective is to minimize the maximum total group cost ($mtgc$) 
to ensure that each group as a whole is not too far from the facility.  
More specifically, given a true profile $r$ and a facility location $y$, 
\begin{align*}
    mtgc(y, r)=\max_{1\leq j \leq m}\left\{\sum_{i\in {G_j}}c(y,x_i)\right\}. 
\end{align*}
Our second group-fair cost objective is to minimize the maximum average group cost ($magc$). Therefore, we have
\begin{align*}
    magc(y, r)=\max_{1\leq j \leq m}\left\{\frac{\sum_{i\in {G_j}}c(y,x_i)}{|G_j|}\right\},
\end{align*}
and we hope to ensure that each group, on average, is not too far from the facility. 
We measure the performance of a mechanism $f$ by comparing the objective that $f$ achieves and the objective achieved by the optimal solution.
%For a mechanism $f$, i
If there exists a number $\alpha$ such that for any profile $r$, 
the output from $f$ is within $\alpha$ times the objective achieved by the optimal solution, 
then we say the approximation ratio of $f$ is $\alpha$.

\subsection{Alternative Group-Fair Social Objectives}\label{sec::agso}
In addition to the objectives defined earlier, we can also consider the following natural objectives 
for group-fair facility location problems: %that model some from of equity: 
\begin{align*}
(a) \max_{1\leq j \leq m}\{h_j\} - \min_{1\leq j \leq m}\{h_j\}  \quad \text{ and } \quad  (b) \frac{\max_{1\leq j \leq m}\{h_j\}}{\min_{1\leq j \leq m}\{h_j\}}, 
\end{align*}
where $h_j$ is a function that can be (i) $\sum_{i\in {G_j}}c(y,x_i)$,  (ii) $\frac{\sum_{i\in {G_j}}c(y,x_i)}{|G_j|}$ or (iii) $\max_{i\in G_j}c(y,x_i)$, 
which implies that each of (a) and (b) has three different group-fair objectives.
In general, both (a) and (b) capture the difference between groups in terms of 
difference and ratio, respectively, under the desirable $h_i$. 
For objectives under type (a), it can be seen as a group envy extended from individual envy works~\cite{Cai:2016aa,DBLP:conf/aaim/ChenFLD20}, and type (a) with function (i) is exactly measure (7) in~\cite{marsh1994equity}. For objectives under type (b), it can be seen as a group envy ratio extended from previous individual envy ratio studies~\cite{DBLP:conf/aaim/LiuDCFN20,ding2020facility}.
%{\color{red} mention connection here}

Surprisingly, we show that any 
deterministic strategyproof mechanism for those objectives 
have unbounded approximation ratios. 

\begin{theorem}\label{thm:first3unbounded}
    Any deterministic strategyproof mechanism does not have a finite approximation ratio for minimizing the three different objectives (i), (ii), and (iii) under (a).
\end{theorem}

\begin{proof}
    We prove this theorem by contradiction. Assume that there exists a  deterministic strategyproof mechanism $f$ with a finite approximation ratio for those objective functions. Consider a profile $r$ with one agent in group $G_1$ at $0$ and one agent in group $G_2$ at $1$. The optimal location is $\frac{1}{2}$ for 
    all (i), (ii), (iii) under (a)
    %all expressions of the first objective function
    and all of their objective values are $0$. Therefore, $f$ has to output $\frac{1}{2}$, otherwise the objective value for $f$ is not equal to $0$ and then the approximation ratio is a non-zero number divided by zero, which is unbounded. 
    
    Then consider another profile $r'$ with one agent in group $G_1$ at $0$ and one agent in group $G_2$ at $\frac{1}{2}$. The optimal location is $\frac{1}{4}$ for 
    all (i), (ii), (iii) under (a)
    %all expressions of the first objective function 
    and all of their objective values are $0$ and then $f$ has to output $\frac{1}{4}$. In that case, given the profile $r'$, the agent at $\frac{1}{2}$ can benefit by misreporting to 1, thus moving the facility location from $\frac{1}{4}$ to $\frac{1}{2}$. This is a contradiction to the strategyproofness .
\end{proof}

\begin{theorem}\label{thm:last3unbounded}
    Any deterministic strategyproof mechanism does not have a finite approximation ratio for minimizing the three different objectives (i), (ii), and (iii)  under (b). 
\end{theorem}

\begin{proof}
    We will reuse the profiles in Theorem \ref{thm:first3unbounded} and prove this theorem by contradiction. Assume that there exists a deterministic  strategyproof mechanism $f$ with a finite approximation ratio. Consider a profile $r$ with one agent in group $G_1$ at $0$ and one agent in group $G_2$ at $1$. The optimal location is $\frac{1}{2}$  and without loss of generality, we assume that $f(r)=\frac{1}{2}+\epsilon$, $0\le\epsilon<\frac{1}{2}$.

    Then consider another profile $r'$ with one agent in group $G_1$ at $0$ and one agent in group $G_2$ at $\frac{1}{2}+\epsilon$. Then $f$ can output any location except $0$ and $\frac{1}{2}+\epsilon$, otherwise 
    all (i), (ii), (iii) under (b)
    %all three expressions of the last objective function 
    are unbounded but the optimal objective value is $1$, then the approximation ratio is unbounded. In that case, given the profile $r'$, the agent in group $G_2$ can benefit by misreporting to $1$, in contradiction to strategyproofness.
\end{proof}

Notice that \cite{marsh1994equity} consider 20 group-fair objectives in total. %\textcolor{red}{(add exact number)}.
However, using a similar technique as in the proof of Theorem~\ref{thm:first3unbounded}, we can show that 
any deterministic strategyproof mechanism does not have a finite approximation ratio for all of the objectives mentioned in \cite{marsh1994equity} except measure (1), which is the $mtgc$ in our paper (we will explore this objective in Section~\ref{sec::mtgc}). 
%in a similarly way as the proof of Theorem~\ref{thm:first3unbounded} by making a contradiction. 
The main reason is that for the other objectives in \cite{marsh1994equity} containing the form such as one group cost minus the other group cost (e.g., objective type (a) with function (i) we mentioned earlier), we can easily construct  profiles similar to those in the proof of  Theorem~\ref{thm:first3unbounded} where the optimal value is $0$.
%{\color{red} Be specific about objectives ... which ones are ours and which ours are theirs}

\section{Mechanism Design for Group-Fair Objectives}
In this section, we consider two group fair objectives, the maximum total group cost ($mtgc$) and the maximum average group cost ($magc$). 
First, we consider three classical strategyproof mechanisms proposed by \cite{Procaccia:2009aa}, which are independent of group information.

\begin{namedthm*}{Median Deterministic Mechanism}[MDM] Given $x_1\leq x_2 \leq x_3 ... \leq x_n$, put the facility at $y=x_{\lceil\frac{n}{2}\rceil}$.
\end{namedthm*}

\begin{namedthm*}{Leftmost Deterministic Mechanism}[LDM] Given $x_1\leq x_2 \leq x_3 ... \leq x_n$, put the facility at $y=x_1$.
\end{namedthm*}

\begin{namedthm*}{Randomized Mechanism}[RM] 
    Given $x_1\leq x_2 \leq x_3 ... \leq x_n$, put the facility at $x_1$ with probability 1/4, $x_n$ with probability 1/4, $\frac{x_1+x_n}{2}$ with probability 1/2.
\end{namedthm*}

Because none of the three mechanisms above leverages group information, 
their strategyproofness still holds in our model. However, 
for some group-fair objectives, they might perform poorly. 
%they have bad performance in some group-fair objectives.
Thus, we propose the following deterministic and randomized mechanisms, %from both deterministic and randomized perspectives, 
which depend on group information, in which case strategyproofness is no longer obvious.

%{\color{red} we need to say why we introduce them  }

\begin{namedthm*}{Majority Group Deterministic Mechanism}[MGDM] 
    Let $g \in \arg\max_{1\leq j\leq m}|G_j|$, put facility $y$ at the median of group $G_g$ (break ties by choosing the smallest index).
\end{namedthm*}

%{\color{red} issues: (1) the reviewers just said it is unfair; why do we call it fair? this would just kill the paper if a reviewer sees this  
%(2) we also need to include the explanation of we gave about this might seems unfair; 
%(3) also we need to say what happen if we put the facility at the median of other groups}

%First, we prove their strategyproofness.

\begin{proposition}\label{pro:MGDM_sp}
    MGDM is strategyproof.
\end{proposition}

\begin{proof}
    For agents not in group $G_g$, they cannot manipulate the output of the mechanism, thus they have no incentive to misreport their locations. 
    For agents in group $G_g$, they cannot change the output unless they misreport their locations to the other side of $y$, %{\color{red} across?}, 
    which makes the facility move farther away from them. Thus, in this case, if the agent misreports, the facility will move farther away from the agent.
    Hence, they have no incentive to misreport their locations. Therefore, MGDM is strategyproof.
\end{proof}

Our goal is to consider a facility’s location that strikes a balance between unfairness across all of the groups while accounting for potential misreporting. While MGDM might focus on a group of the agents when locating a facility, it implicitly considers the unfairness induced by other groups of the agents given their preferences in our later approximation analyses. Without such a consideration, the mechanisms can perform badly and induce prohibitively high costs or unfairness for groups. 
Notice that we can alternatively consider other deterministic mechanisms such as choosing the median of all group median agents or choosing the median of another group, but they cannot guarantee  strategyproofness or achieve better approximation ratios (i.e., counterexamples are given in the in Appendix) for both objectives. 

%Although MGDM seems unfair, we notice that any strategyproof mechanisms can be unfair for different groups of agents under some instances. Our goal is to identify a facility’s location that strikes a balance between unfairness across all of the groups while accounting for potential misreporting. We also take some other reasonable deterministic mechanisms such as choosing the median of all group median agents or choosing the median of other groups into considerations, but they cannot guarantee the strategyproofness or achieve better approximation ratios (counter examples will be shown in Appendix) {\color{red} for what objectives?}. 

\begin{namedthm*}{Narrow Randomized Mechanism}[NRM]
    Let $M$ be a set of median agents of all groups (choose the left one if there are two median agents in the group) and let $ml=\arg\min_{i\in M}\{x_i\}$ and $mr = \arg\max_{i\in M}\{x_i\}$, put the facility at $x_{ml}$ with probability 1/4, $x_{mr}$ with probability 1/4, and $\frac{x_{ml}+x_{mr}}{2}$ with probability 1/2.
\end{namedthm*}

\begin{proposition}\label{pro::NRM_sp}
    NRM is strategyproof.
\end{proposition}

\begin{proof}
    By definition we have $x_{ml}\leq x_{mr}$. 
    For agents whose locations are outside of $[x_{ml},x_{mr}]$, they can only change the facility location by misreporting their locations to the other side of their own group's median point, which might further make one of $x_{ml}$ and $x_{mr}$ (together with the midpoint) move farther away. Therefore, they have no incentive to misreport. %their locations.

    For agents whose locations are in the interval $[x_{ml},x_{mr}]$, they can only change the facility location by misreporting their locations to the other side of their own group's median point, which might further make either $x_{ml}$ or $x_{mr}$ move farther away. But the midpoint of two medians may be closer to them. Suppose that the median point moves by $\Delta$ after an agent misreports and that we get a new facility location distribution $Y'$, then the midpoint of two medians will move by $\frac{1}{2}\Delta$ and the cost of the misreporting agent satisfies
%\begin{align*}
    %& 
    $    \mathbb{E}(d(Y',x_i))-\mathbb{E}(d(Y,x_i)) 
    \leq  \frac{1}{4} \Delta + \frac{1}{2}(-\frac{1}{2}\Delta) + \frac{1}{4}(0) = 0. $
%\end{align*}
    Therefore, they have no incentive to misreport their locations. Hence, NRM is strategyproof. 
\end{proof}

When designing the NRM, 
we observe that for any profile, the optimal solutions of both group-fair objectives are in $[x_{ml},x_{mr}]$ since all group total (average) costs increase from $x_{ml}$ to the left and from $x_{mr}$ to the right. Then putting the facility outside the interval we mentioned above with a certain probability will only hurt the mechanism's performance. Therefore, it would be better to design a randomized mechanism which only puts the facility in $[x_{ml},x_{mr}]$. 
Based on this fact, we use the same probabilities as RM to guarantee the strategyproofness, but use $x_{ml}$ and $x_{mr}$ 
instead of $x_1$ and $x_n$%the endpoints 
 to achieve better performance.

In the following subsections, 
we show the approximation ratios of these mechanisms and provide lower bounds for minimizing the two group-fair objectives. %, respectively.

\subsection{Maximum Total Group Cost}\label{sec::mtgc}

%{\color{red} move the table up here and add a sentence of two describing the results; so that people don't have to look at the proofs if they don't want to} 

In this subsection, we focus on minimizing the maximum total group cost, Table \ref{tab:mtgc_sum} summaries the results.
\begin{table}[htb!]
    \centering
    \begin{tabular}{c|c|c}
    \toprule
        Mechanism & Approximation Ratio & Lower Bound\\
    \midrule
        MDM & $\Omega(m)$ &\multirow{3}*{2} \\
        LDM & $\Omega(n)$ &\\
        MGDM & $3$  & \\
        \hline
        RM & $\Omega(n)$  &\multirow{2}*{1.5}\\
        NRM & $\Omega(n)$ & \\
    \bottomrule
    \end{tabular}
    \caption{Summary of Results for the $mtgc$.}
    \label{tab:mtgc_sum}
\end{table}

We first provide upper bounds for the considered deterministic mechanisms discussed earlier. 
It turns out that existing mechanisms, MDM and LDM, do not perform well for the $mtgc$ objective. 

\begin{proposition}\label{pro:MDMub}
    MDM has an approximation ratio of $\Omega(m)$ for minimizing the $mtgc$.
\end{proposition}

\begin{proof}
    Consider the case (see Figure \ref{fig:ub_mtgc_m}) with $m$ agents and $m$ groups $G_1,...,G_m$ where all agents at $0$ belong to different groups, and $m$ agents at $1$ belong to $G_1$. MDM puts the facility at $0$ achieving the $mtgc$ of $m$, but putting the facility at%{\color{red} putting the facility at?} 
    location $1$ can achieve the optimal $mtgc$ value $1$. Hence, the approximation ratio of MDM for the $mtgc$ is at least $m$.
\end{proof}

\begin{figure}[htp!]
    \centering
    \includegraphics[width=0.4\linewidth]{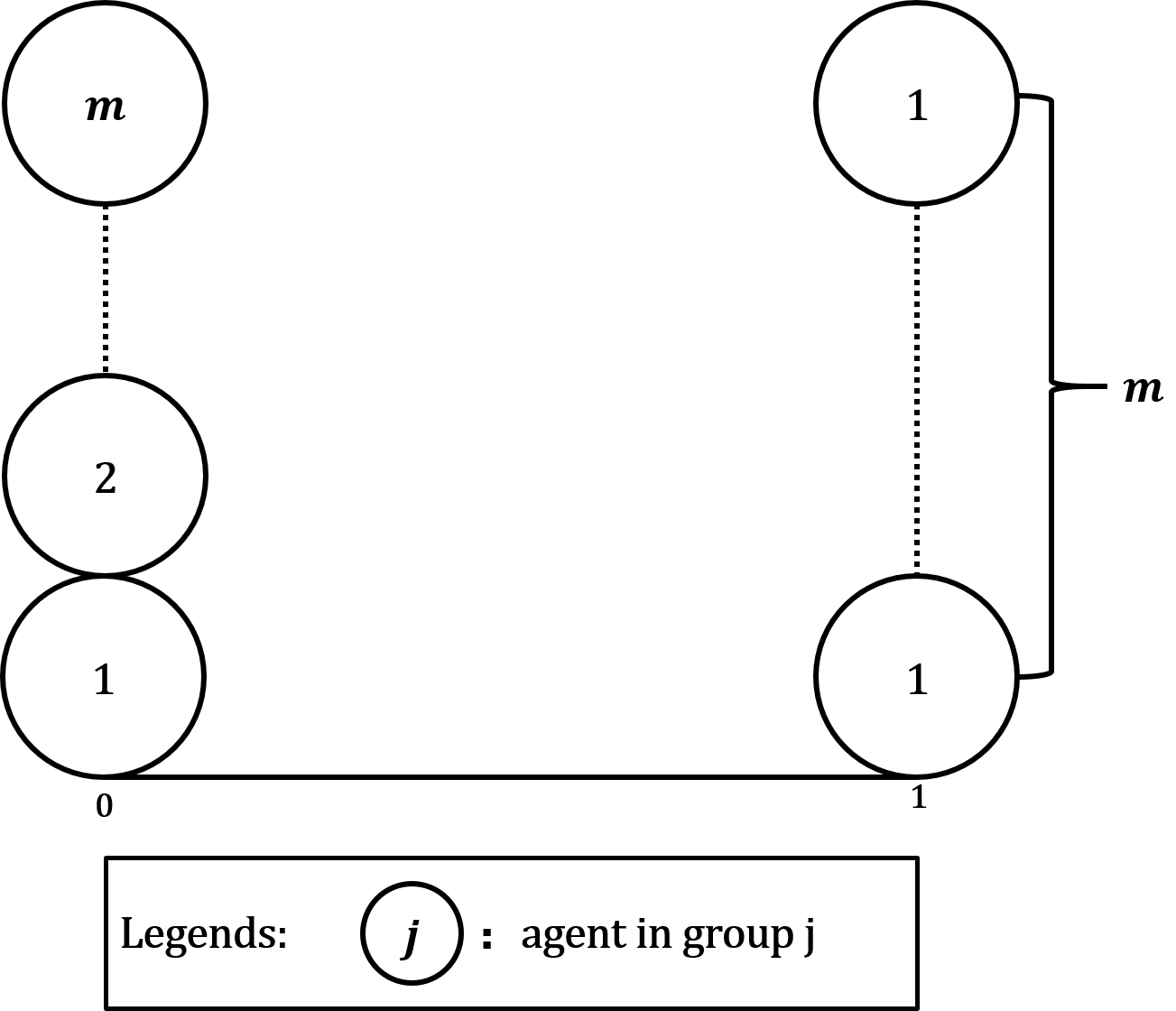}
    \caption{Profile used in the proof of Proposition~\ref{pro:MDMub}.}
    \label{fig:ub_mtgc_m}
\end{figure}

\begin{proposition}\label{pro:LDMub}
    LDM has an approximation ratio of $\Omega(n)$ for minimizing the $mtgc$.
\end{proposition}

\begin{proof}
    Consider the case with one agent at $0$ and $n-1$ agents at $1$ where all of them belong to the same group. LDM puts the facility at $0$ achieving the $mtgc$ of $n-1$, but putting the facility at location $1$ can achieve the optimal $mtgc$ value $1$. Hence, the approximation ratio of LDM for the $mtgc$ is at least $n-1$.
\end{proof}

Next, we show that our mechanism, MGDM, that leverages group information has a good constant approximation ratio. 

\begin{theorem}\label{thm:GFDMub}
    MGDM has an approximation ratio of $3$ for minimizing the $mtgc$.
\end{theorem}

\begin{proof}
    Given any profile $r$, let $y$ be the output of MGDM, $y^*$ be the optimal location and without loss of generality we assume that $y < y^*$ and $mtgc(y,r)$ is achieved by $G_{g'}$. 
    Then from 
    \begin{align*}
    mtgc(y^*,r) &= \max_j \left\{\sum_{i\in G_j}d(y^*,x_i)\right\} 
    \geq \sum_{i\in G_{g'}}d(y^*,x_i),
    \end{align*}
    we have
    \begin{align*}
        &mtgc(y,r)-mtgc(y^*,r) 
        \leq  \sum_{i\in G_{g'}}c(y,x_i) - \sum_{i\in G_{g'}}c(y^*,x_i)\\
        &=\sum_{i\in G_{g'}}|y-x_i|-\sum_{i\in G_{g'}}|y^*-x_i| 
        \leq |G_{g'}|(y^*-y). 
    \end{align*}
    By a simple transformation, we further have 
    \begin{align*}
        mtgc(y,r)\leq & mtgc(y^*,r)+|G_g'|(y^*-y)\\
        \leq&  mtgc(y^*,r)+|G_g|(y^*-y). 
    \end{align*}
    Moreover, $mtgc(y^*,r)$ is at least $\frac{1}{2}|G_g|(y^*-y)$ because there are at least $\frac{|G_g|}{2}$ agents on the left side of $y$. Then we have the  approximation ratio
    \begin{align*}
        \rho &= \frac{mtgc(y,r)}{mtgc(y^*,r)}
        \leq \frac{mtgc(y^*,r)+|G_g|(y^*-y)}{mtgc(y^*,r)}\\
        &\leq \frac{\frac{1}{2}|G_g|(y^*-y)+|G_g|(y^*-y)}{\frac{1}{2}|G_g|(y^*-y)}
        = 3. 
    \end{align*}
\end{proof}

\begin{figure}
    \centering
    \includegraphics[width=0.5\linewidth]{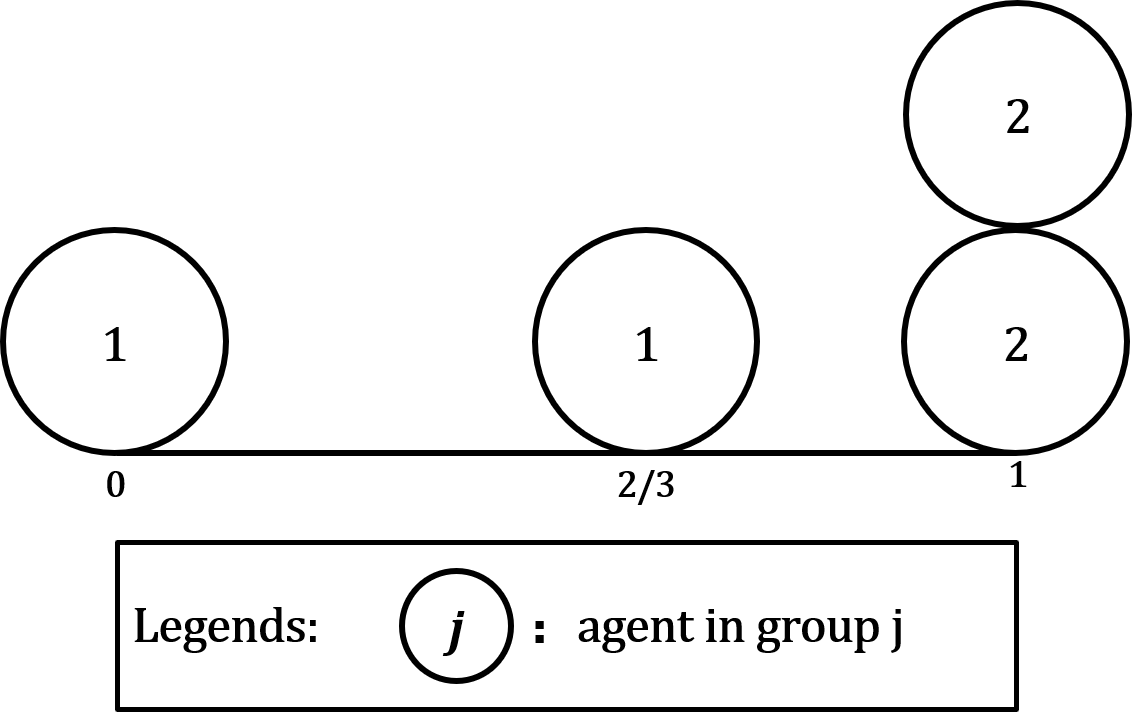}
    \caption{Profile used in Example~\ref{exam:mtgc_3}.}
    \label{fig:ub_mtgc_3}
\end{figure}

\vspace{-1em}

In fact, the following example shows that the bound given in Theorem~\ref{thm:GFDMub} is tight.

\begin{example}\label{exam:mtgc_3}
    Consider the case (see Figure~\ref{fig:ub_mtgc_3}) with one agent in group $G_1$ at 0, one agent in group $G_1$ at 2/3, and two agents in group $G_2$ at 1. MGDM puts the facility at $0$ achieving the $mtgc$ of $2$, but location $2/3$ can achieve the optimal $mtgc$ value $2/3$. Hence, the approximation ratio of MGDM for the $mtgc$ is at least $3$.
\end{example}

To complement our upper bound results, we provide a lower bound for this objective. 

\begin{theorem}\label{thm:MTGClb}
    Any deterministic strategyproof mechanism has an approximation ratio of at least $2$ for minimizing the $mtgc$.
\end{theorem}

\begin{proof}
    Assume that there exists a strategyproof mechanism $f$ where the approximation ratio is smaller than $2$. 
    Consider a profile where one agent in group $G_1$ is at $0$ and one agent in group $G_2$ is at $1$. 
    Without loss of generality, we assume that $f(r) = \frac{1}{2}+\epsilon$ and $\epsilon\geq 0$. 
    Now, consider the profile $r'$ where one agent in group $G_1$ is at $0$ and one agent in group $G_2$ is at $\frac{1}{2}+\epsilon$. 
    The optimal solution puts the facility at $1/4 + \epsilon/2$, which has a maximum total group cost of $1/4 + \epsilon/2$. 
    If the mechanism is to achieve an approximation ratio better than 2, the facility must be placed in $(0, 1/2+\epsilon)$. 
    In that case, given the profile $r'$, the agent in group $G_2$ can benefit by misreporting to 1, thus moving the facility to $1/2 + \epsilon$, in contradiction to strategyproofness.
\end{proof}

Next, we investigate whether the considered randomized mechanisms can help to improve the approximation ratios. 
Unfortunately, these mechanisms do not perform well for this objective. 
For completeness, we provide a lower bound for any strategyproof randomized mechanisms. 

%Then we consider randomized mechanisms.

\begin{proposition}\label{pro:RMub}
    RM and NRM have an approximation ratio of $\Omega(n)$ for minimizing the $mtgc$.
\end{proposition}

\begin{proof}
    Consider the case where one agent in $G_1$ is at 0, one agent in $G_2$ is at 1, and $n-2$ agents in $G_3$ are at $1/2$. Both RM and NRM put the facility at $0$ with probability $1/4$, $1$ with probability $1/4$, $1/2$ with probability $1/2$, achieving the $mtgc$ of $(n-1)/4$, but putting the facility at location $1/2$ can achieve the optimal $mtgc$ value $1/2$. Hence, the approximation ratio of RM and NRM for the $mtgc$ is at least $(n-1)/2$.
\end{proof}

% \begin{proposition}\label{pro:NRMub}
%     NRM has an approximation ratio of $\Omega(n)$ for minimizing the $mtgc$.
% \end{proposition}

\begin{theorem}\label{thm:MTGClb2}
    There does not exist any strategyproof randomized mechanism with an approximation ratio less than 3/2 for minimizing the $mtgc$.
\end{theorem}

\begin{proof}
    Consider a profile where one agent in $G_1$ is at 0 and one agent in $G_2$ is at 1. In this case, the maximum total group cost is equivalent to the maximum cost, and we can use a similar argument as the lower bound of the maximum cost in \cite{Procaccia:2009aa} to prove the theorem.
\end{proof}

\subsection{Maximum Average Group Cost}

%{\color{red} move the table up here and add a sentence of two describing the results; so that people don't have to look at the proofs if they don't want to}

In this subsection, we focus on minimizing the maximum average group cost, Table \ref{tab:magc_sum} summaries the results.
\begin{table}[htb!]
    \centering
    \begin{tabular}{c|c|c}
    \toprule
        Mechanism & Approximation Ratio & Lower Bound\\
    \midrule
        MDM & $3$ &\multirow{3}*{2} \\
        LDM & $\Omega(n)$ &\\
        MGDM & $3$  & \\
        \hline
        RM & $\Omega(n)$  &\multirow{2}*{1.5}\\
        NRM & $2$ & \\
    \bottomrule
    \end{tabular}
    \caption{Summary of Results for the $magc$.}
    \label{tab:magc_sum}
\end{table}

We first provide upper bounds for the  considered deterministic mechanisms. Surprisingly, both MDM and MGDM give an approximation ratio of 3.

\begin{theorem}\label{thm:MAGCub}
    MDM has an approximation ratio of $3$ for minimizing the $magc$.
\end{theorem}

\begin{proof}
    Given any profile $r$, let $y$ be the output of MDM, $y^*$ be the optimal location and without loss of generality we assume that $y < y^*$ and $magc(y,r)$ is achieved by $g'$. 
    Then by 
     \begin{align*}
     magc(y^*,r) &= \max_j\left\{\frac{\sum_{i\in G_j}d(y,x_i)}{|G_j|}\right\}  
     \geq \frac{\sum_{i\in G_{g'}}d(y,x_i)}{|G_{g'}|},
     \end{align*}
     we have
     \begin{align*}
        & magc(y,r)-magc(y^*,r) \\
        \leq& \sum_{i\in G_{g'}}c(y,x_i)/|G_{g'}| - \sum_{i\in G_{g'}}c(y^*,x_i)/|G_{g'}|\\
        =&\frac{\sum_{i\in G_{g'}}|y-x_i|}{|G_{g'}|}-\frac{\sum_{i\in G_{g'}}|y^*-x_i|}{|G_{g'}|}\\
        =&\frac{\sum_{i\in G_{g'}}(|y-x_i|-|y^*-x_i|)}{|G_{g'}|} 
        \leq y^*-y. 
    \end{align*}
    
     Therefore, $magc(y,r)\leq magc(y^*,r)+(y^*-y)$ and we also have $magc(y^*,r)\geq \frac{1}{2}(y^*-y)$ because there is at least one group with at least half group members on the left of $y$. Then we have the approximation ratio
    \begin{align*}
        \rho &= \frac{magc(y,r)}{magc(y^*,r)} 
        \leq \frac{magc(y^*,r)+(y^*-y)}{magc(y^*,r)}\\
        &\leq \frac{\frac{1}{2}(y^*-y)+(y^*-y)}{\frac{1}{2}(y^*-y)}
        = 3. 
    \end{align*}
\end{proof}

For LDM, we can reuse the proof of Proposition \ref{pro:LDMub} since all agents are in the same group in the proof, minimizing the $mtgc$ is equivalent to minimizing the $magc$. 

\begin{corollary}
    LDM has an approximation ratio of $\Omega(n)$ for minimizing the $magc$.
\end{corollary}

For MGDM, we can also use a similar argument as in the proof of Theorem~\ref{pro:MDMub} to show that MGDM can achieve an upper bound of $3$ since for any profile $r$, there are at least half of the members who are in the largest group and on the left-hand side of the output of MGDM.

\begin{corollary}
    MGDM has an approximation ratio of $3$ for minimizing the $magc$.
\end{corollary}

Next, we investigate the lower bound. We can reuse the proof of Theorem \ref{thm:MTGClb} to show the same lower bound since there is only one agent in each group in the proof, minimizing the maximum average group cost is equivalent to minimizing the maximum total group cost.

\begin{corollary}
\label{thm:MAGClb}
Any deterministic strategyproof mechanism has an approximation 
ratio of at least $2$ for minimizing the $magc$.
\end{corollary}

We now investigate whether the considered randomized mechanisms can help to improve the approximation ratios. While the existing randomized mechanism does not, our mechanism, NRM, improves the approximation ratio to 2. 

%In contrast to the $mtgc$ where MGDM makes a significant improvement for the approximation ratio, we can observe that NRM provide a much better approximation ratio for the maximum average group cost objective. 

\begin{proposition}\label{pro:RMmagc}
    RM has an approximation ratio of $\Omega(n)$ for minimizing the $magc$.    
\end{proposition}

\begin{proof}
    Consider the case with one agent at 0, one agent at 1, and $n-2$ agents at $1/2$ where all of them are in the same group. RM puts the facility at $0$ with probability $1/4$, $1$ with probability $1/4$, $1/2$ with probability $1/2$, achieving the $magc$ of $(n+2)/(4n)$, but location $1/2$ can achieve the optimal $mtgc$ value $1/n$. Hence, the approximation ratio of NRM for the $mtgc$ is at least $(n+2)/4$.
\end{proof}

\begin{theorem}\label{thm:MAGCub2}
    NRM has an approximation ratio of $2$ for minimizing the $magc$.
\end{theorem}

\begin{proof}
Given any profile $r$, if $x_{ml}=x_{mr}$, NRM only puts the facility at $y=x_{ml}=x_{mr}$, which is the optimal location. Therefore, we only consider the case where $x_{ml}\neq x_{mr}$. Without loss of generality, we assume that $x_{ml} < x_{mr}$ and let $y^*$ be the optimal solution.
    
Without loss of generality, we assume that $y^*\in[x_{ml}, \frac{x_{ml}+x_{mr}}{2}]$. If $magc(x_{ml},r)$ is achieved by $G_p$ and since there are at most $|G_p|$ members of group $G_p$ on the right of $y^*$, we have
\begin{align*}
    magc(x_{ml},r) 
    \leq& \frac{\sum_{i\in G_p}c(x_i,y^*)}{|G_p|} + (y^*-x_{ml})\\
    \leq& magc(y^*,r) + (y^* - x_{ml}). 
\end{align*}
Similarly, if $magc(\frac{x_{ml}+x_{mr}}{2},r)$ is achieved by $G_q$ and $magc(x_{mr}, r)$ is achieved by $G_s$, we have
\begin{align*}
    magc(\frac{x_{ml}\!+\!x_{mr}}{2},r) 
    &\!\leq\! \frac{\sum_{i\in G_q}c(\!x_i,y^*\!)}{|G_q|}\! +\! (\frac{x_{ml}\!+\!x_{mr}}{2}\! -\! y^*\!)\\
    &\leq magc(y^*,r)\! +\! (\frac{x_{ml}+x_{mr}}{2}\!-\!y^*), 
\end{align*}
and
\begin{align*}
    magc(x_{mr},r) 
    \leq& \frac{\sum_{i\in G_s}c(x_i,y^*)}{|G_s|} + (x_{mr} - y^*)\\
    \leq& magc(y^*,r) + (x_{mr} - y^*). 
\end{align*}

Therefore, the approximation ratio is 
\begin{align*}
    % \rho = & \frac{\frac{1}{4}magc(x_{ml}, r)+\frac{1}{2}magc(\frac{x_{ml}+x_{mr}}{2}, r)}{magc(y^*, r)}
    % +  \frac{\frac{1}{4}magc(x_{mr}, r)}{magc(y^*, r)}\\
    \rho = & \frac{\frac{1}{4}magc(x_{ml}, r)+\frac{1}{2}magc(\frac{x_{ml}+x_{mr}}{2}, r)+\frac{1}{4}magc(x_{mr}, r)}{magc(y^*, r)}\\
    \leq & \frac{\frac{1}{4}(magc(y^*, r)\!+\!(y^*\!-\!x_{ml}))\!}{magc(y^*,r)}\\
    &+\frac{\!\frac{1}{2}(magc(y^*, r)\!+\!(\frac{x_{ml}+x_{mr}}{2}\!-\!y^*))}{magc(y^*,r)}\\
    &+\frac{\frac{1}{4}(magc(y^*,r)+(x_{mr}-y^*))}{magc(y^*,r)}
    = 1+\frac{\frac{1}{2}(x_{mr}-y^*)}{magc(y^*,r)}. 
\end{align*}
Furthermore, if the agent $mr$ is in group $G_R$, then
there are at least $\frac{|G_R|}{2}$ agents of group $G_R$ in $[x_{mr},\infty)$ since $mr$ is median agent. Therefore, we have
\begin{align*}
    magc(y^*,r) &\geq 
    \frac{\sum_{i\in G_R,x_i>x_{mr}}(x_{mr}-y^*)}{|G_R|} 
    > \frac{x_{mr}-y^*}{2}. 
\end{align*}
Thus, the approximation ratio is 
\begin{align*}
    \rho &= 1+\frac{\frac{1}{2}(x_{mr}-y^*)}{magc(y^*,r)} 
         \leq  1+\frac{\frac{1}{2}(x_{mr}-y^*)}{\frac{1}{2}(x_{mr}-y^*)}
         = 2.
\end{align*}
\end{proof}

\begin{figure}
    \centering
    \includegraphics[width=0.5\linewidth]{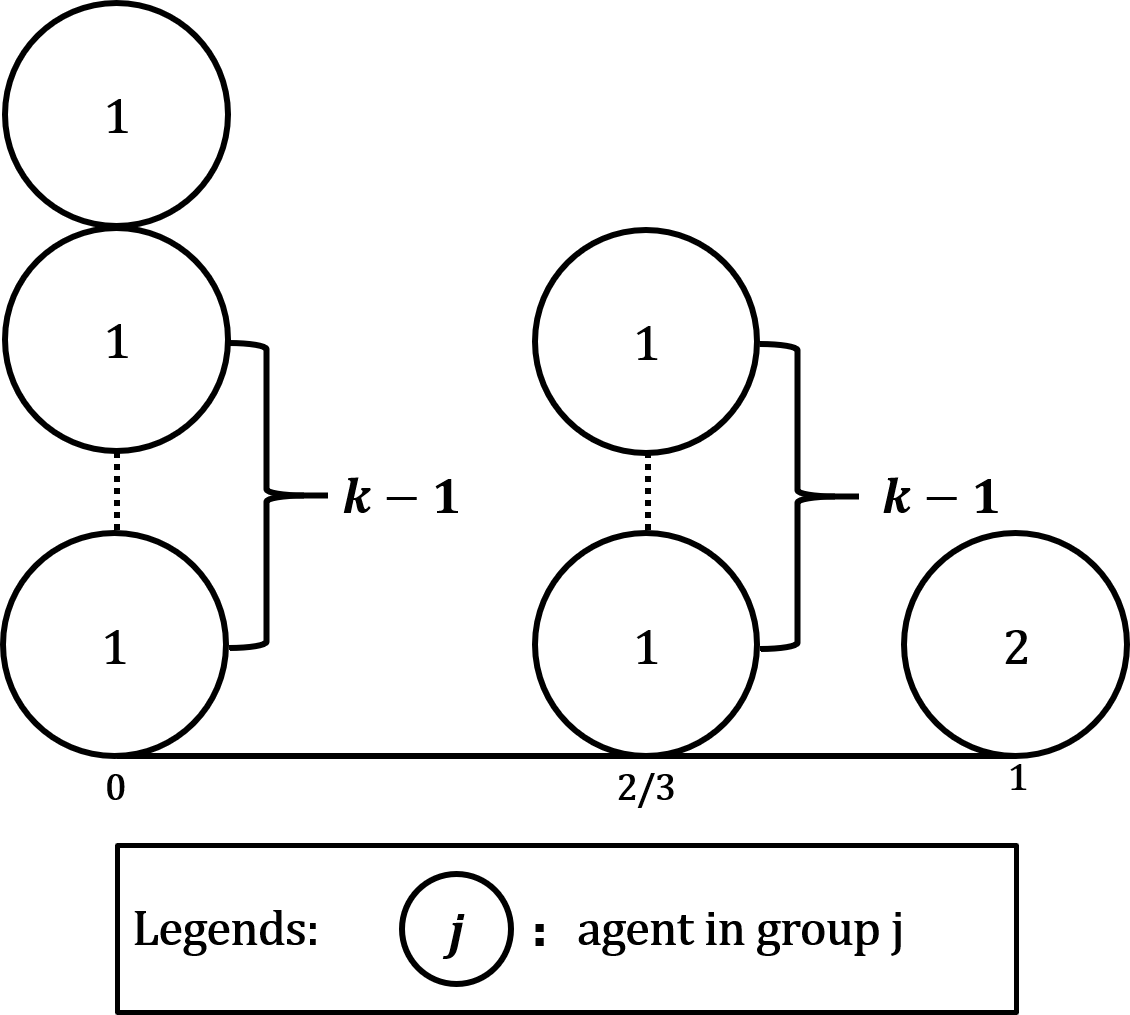}
    \caption{Profile used in Example~\ref{exam:magc_32}.}
    \label{fig:ub_magc_3}
\end{figure}

The following example shows that the bounds given in Theorem~\ref{thm:MAGCub} and Theorem~\ref{thm:MAGCub2} are tight.

\begin{example}\label{exam:magc_32}
    Consider the case (see Figure~\ref{fig:ub_magc_3}) where $k$ agents in group $G_1$ are at 0, $k-1$ agents in group $G_1$ are at $\frac{2}{3}$, and one agent in group $G_2$ is at 1. MDM puts the facility at $0$ achieving the $magc$ of $1$, NRM puts the facility at 0 with probability 1/4, 1/2 with probability 1/2, 1 with probability 1/4, achieving the $magc$ of $2/3$, but the optimal solution puts the facility at $2/3$ achieving the $magc$ of $1/3$.
\end{example}

Next, we provide a lower bound to complement our upper bound results. 
Note that we can reuse the argument as in the proof of Theorem \ref{thm:MTGClb2}, since each group has only one agent in the proof, minimizing the maximum cost is equivalent to minimizing the $magc$. 

\begin{corollary}\label{RandLower}
    There does not exist any strategyproof randomized mechanism with an approximation ratio less than 3/2 for minimizing the $magc$.
\end{corollary}

\section{Intergroup and Intragroup Fairness}

In this section, we investigate \textit{Intergroup and Intragroup Fairness} (IIF), which not only captures 
fairness between groups but also fairness within groups. 
IIF is an important characteristic to be considered in the social science domain 
when studying fairness in group dynamics (see e.g., \cite{haslam2014social}). 

To facilitate our discussion, given group $g$, profile $r$ and facility location $y$,
let $avgc(r, g, y)$ be the average cost of agents in $g$, 
$maxc(r,g,y)$ be the maximum cost among agents in $g$, 
and $minc(r,g,y)$ be the minimum cost among agents in $g$. 
We define below new group-fair IIF social objectives which measure both intergroup and intragroup fairness: 
\begin{align*}
    IIF_{1}(y, r) =& \max_{1\leq j \leq m}\{avgc(r,G_j,y)\} \\
    &+ \max_{1\leq j \leq m} \{maxc(r,G_j,y)-minc(r,G_j,y)\}
\end{align*}
\begin{align*}
    IIF_{2}(y, r)  = \max_{1\leq j \leq m}\{ & avgc(r,G_j,y) + maxc(r,G_j,y)\\
    & - minc(r,G_j,y)\}. 
\end{align*}

Using $maxc(r,G_j,y)-minc(r,G_j,y)$ to measure intragroup fairness is well justified since this is the max-envy considered for one group in \cite{Cai:2016aa}. For $IIF_{1}$, the intergroup fairness and the intragroup fairness are two separate indicators and they can be achieved by different groups, while for $IIF_{2}$, we combine these two as one indicator of each group. 

The reason we do not use the total group cost in this combined measure is that when the group size is large, the total cost is large but the maximum cost minus minimum cost is just the cost of one agent. Then the total cost will play a major role and intragroup fairness will be diluted, which goes against the purpose of the combined fairness measure. Moreover, since both the values of the average group cost and the max-envy are in $[0,1]$, we combine them directly without normalization.

Given the objectives, our goal is to minimize $IIF_{1}$ or $IIF_{2}$. 

%{\color{red} the name is quite strange ... do you mean k smallest? }

\begin{namedthm*}{kth-Location Deterministic Mechanism}[k-LDM]
    Given $x_1\leq x_2 \leq x_3 ... \leq x_n$, put the facility at $y=x_k$ ($k=1,2,...,n$).
\end{namedthm*}

$k$-LDM can be seen as a class of mechanisms and LDM is one of them ($k=1$). It is well known that putting the facility at the $k$-th agent's location is strategyproof. Thus we focus on the approximation ratios. %{\color{red} show the following result read as for any k?}

\begin{theorem}\label{thm:IIF1ub}
    For any $1\leq k \leq n$, k-LDM has an approximation ratio of 4 for minimizing $IIF_{1}$ and $IIF_2$.
\end{theorem}

\begin{proof}
    First, we focus on $IIF_1$.
    Let $y^*$ be the optimal location and we assume that $y<y^*$ without loss of generality. Then we prove that the approximation ratio
    \begin{align*}
        \rho 
        \leq& \frac{\! \max_{j}\{avgc(r,G_j,y)\} +\! \max_{j} \{maxc(r,G_j,y)\}}{IIF_1(y^*,r)}
        \leq 4. 
    \end{align*}
As Figure~\ref{fig:ub_IIF_move} shows, given any profile $r$ and for each group $j$, we move all the agents to $y^*-maxc(r,G_j,y^*)$ if they are on the left of $y$, and to $y^*+maxc(r,G_j,y^*)$ if they are on the right of $y$. 
Then we will obtain a new profile $r'$ with the approximation ratio $\rho'$. If we prove these movements do not make the optimal solution increase and do not make the mechanism solution decrease, namely $\rho\leq \rho'$, and further show that $\rho'\leq 4$, then we can obtain the approximation ratio of $4$.
    
After the movements, $maxc(r',G_j,y^*) - minc(r',G_j,y^*) = 0$ 
and $avgc(r',G_j,y^*)=maxc(r,G_j,y^*)$ for all group $G_j$. Without loss of generality, suppose that $\max_{j}\{avgc(r',G_j,y^*)\}$ is achieved by group $G_p$. 
Then we have 
    \begin{align*}
        & \! \max_{j}\{avgc(r,G_j,y^*)\} \\
         &+\! \max_{j} \{maxc(r,G_j,y^*)-minc(r,G_j,y^*)\}\\
        \geq& avgc(r,G_p,y^*)+maxc(r,G_p,y^*)-minc(r,G_p,y^*) \\
        \geq& minc(r,G_p,y^*)+maxc(r,G_p,y^*)-minc(r,G_p,y^*)\\
        =& maxc(r,G_p,y^*) = avgc(r',G_p,y^*), 
    \end{align*}
    which implies the optimal solution does not increase after the movements since $maxc(r',G_j,y^*) - minc(r',G_j,y^*) = 0$. 
    Moreover, because $y$ is an agent location and there exists at least one agent on the left of or at $y$, we have $maxc(r',p,y^*)\geq y^*-y$.
    
    For the mechanism solution, neither of $\max_{j}\{avgc(r,G_j,y)\}$ and $\max_{j}$ $ \{maxc(r,G_j,y)\}$ decreases if $r$ changes to $r'$ since no agent moves closer to $y$. 

    Therefore, the approximation ratio $\rho'$ is at most
    \begin{align*}
        & \frac{(y^*\!+\!maxc(r',G_p,y^*)\!-\!y)\!+\!(y^*\!+\!maxc(r',G_p,y^*)-y)}{(y^*+maxc(r',G_p,y^*)\!-\!y^*)}\\
        &= \frac{2(y^*-y)+2maxc(r',G_p,y^*)}{maxc(r',G_p,y^*)} 
        \leq \frac{4(y^*-y)}{y^*-y} = 4.
    \end{align*}
    
    For $IIF_2$, we can use a similar argument as $IIF_1$ since it only focuses on group $p$ and all inequalities are also valid under this objective. 
\end{proof}

% \begin{corollary}
% For any $1\leq k\leq n$, k-LDM has an approximation ratio of 4 for minimizing $IIF_{2}$.
% \end{corollary}

\begin{figure}
    \centering
    \includegraphics[width=0.7\linewidth]{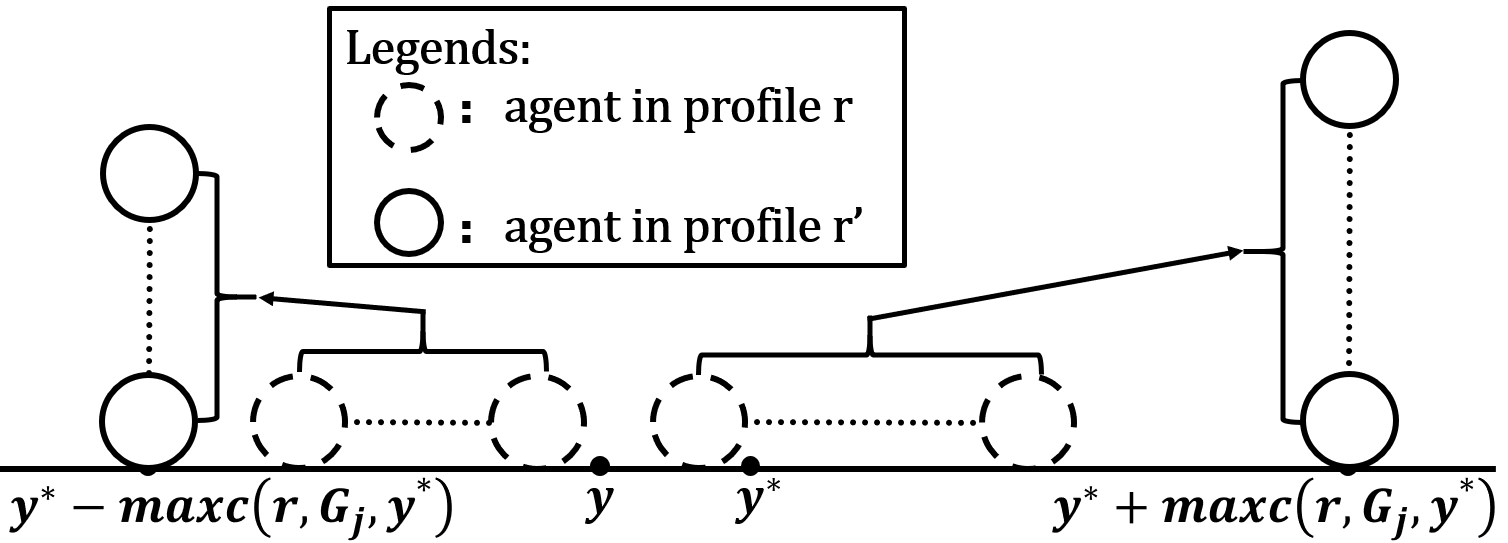}
    \caption{Movement example for group $j$ in the proof of Theorem~\ref{thm:IIF1ub}.}
    \label{fig:ub_IIF_move}
\end{figure}

%Next, we provide lower bounds for the IIF objectives. %we consider lower bounds.
In the following proofs for lower bounds, we use the definition below.

\begin{definition}\cite{Lu:2010aa}
    A mechanism $f$ is \textbf{partial group strategyproof} if and only if a group of agents at the same location cannot benefit even if they misreport their locations simultaneously. More formally, given any profile $r$, let $S\subseteq N$ and all the agents in $S$ have the same location and  $r_S'=\{x_S', g_S\}$ be a profile with the false location reported by agents in $S$. We have $c(f(r), x_S) \leq c(f(r_S', r_{-S}), x_S)$ 
    %or $u(f(r), x_S) \geq u(f(r_S', r_{-S}), x_S)$ 
    where $r_{-S}$ is the profiles reported by all agents except the agents in $S$.
\end{definition}

From the definition, we know that any partial group strategyproof mechanism is also strategyproof. Furthermore, it has been shown that any strategyproof mechanism is also partial group strategyproof in facility location problems \cite{Lu:2010aa}. Thus, we will not distinguish between strategyproof and partial group strategyproof in the following analysis. 

\begin{theorem}\label{IIF1_Lower}
Any deterministic strategyproof mechanism has an approximation ratio of at least $4$ for minimizing $IIF_1$ and $IIF_2$.
\end{theorem}

\begin{proof}
    First, we prove the lower bound of $IIF_1$. Assume for contradiction that there exists a strategyproof mechanism $f$ where the approximation ratio is $4-\epsilon$, $\epsilon>0$. According to the equivalence between strategyproofness and partial group strategyproofness in the facility location problem, $f$ is also partial group strategyproof.
    
    Consider a profile $r$ with one agent in $G_1$ and $\lceil \frac{2}{\epsilon} \rceil$ agents in $G_2$ at $0$, $\lceil \frac{2}{\epsilon} \rceil$ agents in $G_1$ and one agent in $G_2$ at $1$. Assume without loss of generality that $f(r) = \frac{1}{2}+\Delta$, $\Delta >0$. Now, consider the profile $r'$ where one agent in $G_1$ and $\lceil \frac{2}{\epsilon} \rceil$ agents in $G_2$ are at $0$, $\lceil \frac{2}{\epsilon} \rceil$ agents in $G_1$ and one agent in $G_2$ are at $1/2+\Delta$. The optimum is the average of the two locations, namely $1/4 + \Delta/2$, which has an $IIF_{1}$ of $1/4 + \Delta/2$. If the mechanism is to achieve an approximation ratio of $4-\epsilon$, the facility must be placed in $(0,\frac{1}{2}+\Delta)$. In that case, given the profile $r'$, the agents at $\frac{1}{2}+\Delta$ can benefit by misreporting to $1$, thus moving the solution to $1/2 + \Delta$, in contradiction to partial group strategyproofness.
    We can extend this result to $m$ groups by locating $\lceil \frac{2}{\epsilon} \rceil$ agents at $0$ and one agent at $1$ for each group except $G_1$ and $G_2$ in profile $r$, locating $\lceil \frac{2}{\epsilon} \rceil$ agents at $0$ and one agent at $\frac{1}{2}+\Delta$ for each group except $G_1$ and $G_2$ in profile $r'$.
    
    We can reuse the profiles to prove the lower bound of $IIF_2$ since we use the profiles where intergroup and intragroup fairness is achieved by the same group. 
\end{proof}

% \begin{corollary}\label{IIF2_Lower}
% Any deterministic strategyproof mechanism has an approximation ratio of at least $4$ for minimizing $IIF_{2}$.
% \end{corollary}

It is interesting to see that when one only considers intragroup fairness, only the additive approximation is possible. When we combine it with intergroup fairness, a tight multiplicative approximation can be obtained for $IIF_{1}$ and $IIF_{2}$. 
%For $IIF_{2}$, upper bound and lower bound only has a small gap when there are many agents. 

\section{Conclusion}\label{sec::con}
%In this work, we introduce a novel paradigm of approximate group-fair mechanism design without money where we aim to design strategyproof mechanisms that approximately optimize game-theoretic optimization problems, which require input from strategic agents, under group-fair objectives and/or criteria without using money. 
%We demonstrate the paradigm using facility location problems as a case study. In particular, 
\noindent
We study the problem of designing strategyproof mechanisms in group-fair facility location problems (FLPs), where agents are in different groups, under several well-motivated group-fair objectives. We first show that not all group-fair objectives admit a bounded approximation ratio. Then we give the complete results of three classical mechanisms and two new mechanisms for minimizing two group-fair objectives, and we show that it is possible to design strategyproof mechanisms with constant approximation ratios that leverage group information. We also introduce Intergroup and Intragroup Fairness (IIF), which takes both fairness between groups and within each group into consideration. We study two natural IIF objectives and provide a mechanism that achieves a tight approximation ratio for both objectives. 

Naturally, there are many potential future directions for the group-fair FLPs in mechanism design. %and the general approximate group-fair mechanism design without money paradigm. 
For the group-fair FLPs under the considered group-fair objectives, 
an immediate direction is to tighten the gaps between the lower and upper bounds of our results. 
Moreover, one can consider alternative group-fair objectives that are appropriate for specific application domains. 

%% The file named.bst is a bibliography style file for BibTeX 0.99c
\clearpage
\bibliographystyle{named}
\bibliography{ijcai22}

\clearpage
\appendix
\section*{Counter Examples of Other Considered Mechanisms}
1. Choosing the median of all group median agents cannot achieve better approximation ratios for both objectives.
\begin{itemize}
    \item For the $mtgc$, consider a profile with $m$ agents and $m$ groups $G_1,...,G_m$ where all agents at $0$ belong to different groups, and $m$ agents at $1$ belong to $G_1$. The mechanism puts the facility at $0$ achieving the $mtgc$ of $m$, but putting the facility at location $1$ can achieve the optimal $mtgc$ value $1$. Hence, the approximation ratio for the $mtgc$ is at least $m$.
    \item For the $magc$, consider a profile where $k$ agents in group $G_1$ are at 0, $k-1$ agents in group $G_1$ are at $\frac{2}{3}$, and one agent in group $G_2$ is at 1. The mechanism puts the facility at $0$ achieving the $magc$ of $1$, but the optimal solution puts the facility at $2/3$ achieving the $magc$ of $1/3$. Hence, the approximation ratio for the $magc$ is at least $3$.
\end{itemize}
2. Choosing the median of other group cannot achieve better approximation ratios for both objectives. Without loss of generality we assume that we put the facility at the median of group $G_k$.
\begin{itemize}
    \item For the $mtgc$, consider a profile where $|G_k|/2$ agents in group $G_k$ are at 0, $|G_k|/2$ agents in group $G_k$ are at $\frac{2|G_{max}|}{2|G_{max}|+|G_k|}$, and $|G_{max}|$ agents in group $G_{max}$ are at 1. The mechanism puts the facility at $0$ achieving the $mtgc$ of $|G_{max}|$, but putting the facility at location $2/3$ can achieve the optimal $mtgc$ value $\frac{|G_k||G_{max}|}{2|G_{max}|+|G_k|}$. Hence, the approximation ratio for the $mtgc$ is at least $\frac{2|G_{max}|+|G_k|}{|G_k|}\ge 3$.
    \item For the $magc$, consider a profile where $|G_k|/2$ agents in group $G_k$ are at 0, $|G_k|/2$ agents in group $G_k$ are at $2/3$, and $|G_p|$ agents in group $G_p$ ($p\neq k$) are at 1. The mechanism puts the facility at $0$ achieving the $magc$ of $1$, but the optimal solution puts the facility at $2/3$ achieving the $magc$ of $1/3$. Hence, the approximation ratio for the $magc$ is at least $3$.
\end{itemize}

%\section*{Proof of Theorem~\ref{thm:first3unbounded}}

%\section*{Proof of Theorem~\ref{thm:last3unbounded}}

%\section*{Proof of Proposition~\ref{pro:MGDM_sp}}

%\section*{Proof of Proposition~\ref{pro::NRM_sp}}

%\section*{Proof of Proposition~\ref{pro:MDMub}}

%\section*{Proof of Proposition~\ref{pro:LDMub}}

%\section*{Proof of Theorem~\ref{thm:MTGClb}}

%\section*{Proof of Proposition~\ref{pro:RMub}} %and Proposition~\ref{pro:NRMub}}

%\section*{Proof of Theorem~\ref{thm:MTGClb2}}

%\section*{Proof of Proposition~\ref{pro:RMmagc}}

%\section*{Proof of Theorem~\ref{thm:MAGCub}}

%\section*{Proof of Theorem~\ref{thm:MAGCub2}}

%\section*{Proof of Theorem~\ref{IIF1_Lower}}

\end{document}